\documentclass[authoryear,12pt]{elsarticle}

\usepackage{enumitem}
\usepackage{graphicx}
\usepackage{amssymb}
\usepackage{amsthm}
\usepackage{amsfonts}
\usepackage{amsmath}
\usepackage[figuresright]{rotating}
\usepackage{enumerate}
\usepackage[latin1]{inputenc}
\usepackage{tabularx}
\usepackage{natbib}
\usepackage{breakcites}
\usepackage{pifont}
\usepackage{geometry}
\usepackage{txfonts}
\usepackage{color}
\definecolor{slateblue}{RGB}{106,90,205}
\usepackage[bookmarksnumbered=true,bookmarksopen=false, hyperfootnotes=true,  colorlinks=true, linkcolor=slateblue,citecolor=slateblue,urlcolor=slateblue, plainpages=false]{hyperref}
\usepackage{threeparttable}
\usepackage{array}
\usepackage{multirow}
\usepackage{dcolumn}
\usepackage[table]{xcolor}
\usepackage{framed}
\usepackage{xcolor}
\usepackage{natbib}
\usepackage{epstopdf}
\usepackage[capposition=top]{floatrow}
\usepackage{mathrsfs}
\usepackage{lscape}
\usepackage{setspace}
\usepackage{subcaption}
\usepackage{booktabs}
\usepackage{tikz} 
\usetikzlibrary{shapes,arrows} 
\usepackage{accents}
\usepackage{pdflscape}
\usepackage{todonotes}
\usepackage{soul}
\usepackage{makecell}

\soulregister\citet7
\soulregister\citep7
\soulregister\ref7
\soulregister\pageref7

\newtheorem{theorem}{Theorem}
\newtheorem{corollary}{Corollary}

\newcommand{\expectation}[3][0]{%
  \ifcase#1
     E( #2 \mid #3 )
     \or E \bigl( #2 \bigm\vert #3 \bigr)
     \or E \Bigl( #2 \Bigm\vert #3 \Bigr)
     \or E \biggl( #2 \biggm\vert #3 \biggr)
     \or E \Biggl( #2 \Biggm\vert #3 \Biggr)
  \else
     E \left( #2  \;\middle\vert\; #3 \right)
  \fi}

\usepackage{lscape} 

\makeatletter
\def\ps@pprintTitle{%
 \let\@oddhead\@empty
 \let\@evenhead\@empty
 \def\@oddfoot{}%
 \let\@evenfoot\@oddfoot}
\makeatother

\begin{document}

\begin{frontmatter}

\title{Clearing algorithms and network centrality \\ \vspace{12pt}
\small{\today}
}

\author{Christoph Siebenbrunner\fnref{fn1}}
\ead{christoph.siebenbrunner@maths.ox.ac.uk}

\fntext[fn1]{University of Oxford, Mathematical Institute and Institute for New Economic Thinking. All views expressed herein are those of the author and do not necessarily reflect the views of any affiliating organization.}

\date{\today}

\begin{abstract}
I show that the solution of a standard clearing model commonly used in contagion analyses for financial systems can be expressed as a specific form of a generalized Katz centrality measure under conditions that 
correspond to a system-wide shock. This result provides a formal explanation for earlier empirical results which showed that Katz-type centrality measures are closely related to contagiousness. It also allows assessing the assumptions that one is making when using such centrality measures as systemic risk indicators. I conclude that these assumptions should be considered too strong and that, from a theoretical perspective, clearing models should be given preference over centrality measures in systemic risk analyses.
\end{abstract}


\end{frontmatter}

\section{Introduction} \label{sec_Introduction}

The importance of the network of interbank loans for the dynamics of financial crises and systemic risk has been discovered long before the financial crisis of 2008. Early examples of models highlighting this relationship go back to \cite{Rochet1996}, \cite{Freixas2000}, \cite{Allen2000} and \cite{Kiyotaki2002}. \cite{Eisenberg2001} and - independently - \cite{Suzuki2002} have developed a clearing model for obligation networks that has become widely applied for modeling contagion due to default on bilateral loans. The model is agnostic to the nature of the firms that are exposed to each other in principle, but the primary applications usually focus on interbank markets. For this reason I will refer to the entities in the system as banks in this paper. \cite{Elsinger2006} provide one of the first of such applications within an integrated stress testing system, \cite{Elsinger2009}, \cite{Rogers2013} and \cite{Fischer2014} have developed important extensions. \cite{Battiston2012} and \cite{Furfine2003} have introduced alternative models for computing contagion effects that do not rely on the fixed-point argument of the clearing model. \cite{Upper2011} provides a good overview of different applications of such models. \cite{Barucca2016} have sought to unify the literature by introducing a general framework that encompasses many of the aforementioned models as special cases and connecting it to the credit valuation model of \cite{Merton1974}.

Following the theoretical works that showed the connection between financial networks and stability, several authors began studying the statistical properties of empirical networks (\cite{Boss2004}) and the relation of the network structure to contagion effects (\cite{Iori2006}, \cite{Nier2007}). A more recent strand of literature studies the empirical relation between network centrality measures contagiousness. \cite{Kuzubas2014} study characteristics a financial institution that was key to the Turkish financial crisis of 2000 and find increasing trends in several centrality measures prior to the outbreak of the crisis. \cite{Puhr2014} and \cite{Alter2015} find that the Katz centrality and its close cousin (see \cite{Newman2010}), the Eigenvector centrality, have the best explanatory power for contagion losses from a \cite{Eisenberg2001}-type clearing model for the Austrian and the German banking system, respectively. \cite{Kobayashi2013} and \cite{Gauthier2013} obtain similar results in simulations. In this study I demonstrate that these results were indeed to be expected, as the solution of the clearing model converges to a generalized Katz centrality measure as a crisis tends to affect the entire financial system.

At first sight, this result may appear at odds with the reasoning of \cite{Acemoglu2015} or \cite{Tahbaz-Salehi2015}, who argue that - notwithstanding their empirical performance - "\textit{off-the-shelf}" centrality measures such as the Katz centrality are a poor proxy for the results of clearing models from a theoretical perspective, as the latter exhibit non-linearities that typically cannot be captured by those indicators. Indeed, my work highlights the very rigorous assumptions that have to be made in order to be able equate the solution of a clearing model to a Katz centrality measure. This allows a critical appreciation of these assumptions, which leads me to agree with the aforementioned authors as a conclusion.

\section{Clearing model}

The model framework builds on the seminal contribution of \cite{Eisenberg2001}. Consider a financial system composed of a set $\mathscr{N}=\{1,\dots,N-1\}$ of interlinked banks. The linkages among these banks are captured in the bilateral loan matrix L, where $L_{i,j}$ represents the liabilities of bank $i$ towards j $j$. $L$ includes an additional row and column for a sink node which captures liabilities outside the system (e.g. deposits by customers). This specification ensures that the total liabilities of bank $i$ are given by the $i$-th entry of the vector of column sums of the liability matrix. Banks are further endowed with external assets $a$. Table \ref{tab_Definitions} gives an overview of the variables used. 

\textbf{Definitions}

\begin{table}[htp!]
\begin{threeparttable}
\caption{Variable Definitions}
\label{tab_Definitions}
\begin{center}
\begin{tabular}{l l p{8cm} } \\ \hline
 \\
		 \textbf{Definition} & \textbf{[Computation]} & \textbf{Description} \\
        \midrule
         $N \in \mathbb{N}$ & & $N$ is the number of banks in the system considered plus one (for the sink node)\\
		 $a \in \mathbb{R}_+^N$ & & $a_i$ are the external assets of bank $i$\\
         $L \in \mathbb{R}_+^{N\times N}$ & & $L_{ij}$ total liabilities of bank $i$ owed to bank $j$\\
         $l \in \mathbb{R}_+^N$ & $l_i = \sum_j L_{ij}$ & $l_i$ are the total nominal liabilities of bank $i$\\
         $p \in \mathbb{R}_+^N$ & $f(p) = p$ & $p$ is the clearing payment vector of payments that are actually made (as opposed to nominal liabilities)\\

         $C \in [0,1]^{N\times N}$ & $C_{ij} = \begin{cases}
\frac{L_{ji}}{l_j} & \mbox{if } l_j > 0 \\
0 & \mbox{otherwise } \end{cases}$ & If $C_{ij}>0$, it represents the relative share that bank $i$'s claim on bank $j$ has among the total liabilities of bank $j$. Note that if $x$ is a vector of payments made by each bank, $Cx$ gives the value of those payments for those creditors.  \\
         $D \in \{0,1\}^{N\times N}$ & $D_ij(x) = \begin{cases}
1 & \mbox{if } i=j \land a_i + (Cx)_i < l_i \\
1 & \mbox{if } i=j=N \\
0 & \mbox{otherwise } \end{cases}$ & $D$ is a diagonal matrix of default indicators. $D(x)_{ii}=1$ means that bank $i$ is in default under a given payment vector $x$. Note that $Dy$ sets all elements of vector $y$ whose positions belong to non-defaulted banks to zero. Note that the sink node is set to be in default as a convention. \\
    \midrule
    	 $o \in \mathbb{R}_+^N$ & & $o_i$ is the value of bank $i$'s asset before a shock\\
    	 $s \in \mathbb{R}^N$ & & $s_i$ is a shock to bank $i$'s assets\\

		 $r \in [0,1]$ & & $r$ is a recovery rate for assets (if they are worth less than their nominal value)\\
         $m \in (0,1)$ & & $m$ is an interpolation coefficient \\

    \bottomrule  
    
\hline
\end{tabular}
\begin{tablenotes}
\item 
\end{tablenotes}
\end{center}
\end{threeparttable}
\end{table}

The \textbf{balance sheet equation} in this model can be written as:

\begin{equation}
Equity = Assets - Liabilities = a + Cp - l
\end{equation}

A bank is said to be \textbf{insolvent} or \textbf{in default} if its equity is negative. If a bank is insolvent even if all other banks fully repay their liabilities (i.e. $a_i + (Dl)_i < l_i)$) then it is said to be in \textbf{fundamental default}. Note that the value of interbank claims, and thus the equity, depends on the value of interbank payments. If all banks fully repay their liabilities, $p=l$ and the equity can be written as $a + Cl - l$. However, the underlying assumption of the clearing model as introduced by \cite{Eisenberg2001} is that assets have to be used to  repay liabilities, hence insolvent banks cannot fully repay their liabilities. The intuition of a \textbf{clearing payment vector} is that all banks repay the minimum of their total liabilities and the total value of their external assets and the value of their claims on other banks under the clearing payment vector. Furthermore, following the approach of \cite{Rogers2013}, it is assumed that when a bank enters into default, the recovery value after liquidating its assets can be less than the original nominal value.

To formalize the intuition laid out above, consider the following map\footnote{Note that by setting $D_{NN}=1$ for the sink node, I am assuming that this node makes payments even though it has no liabilities. This has no implication for the solutions for the other banks, as these payments do not arrive anywhere. Given that the sink node does not need to make payments within the system (and can generally not be interpreted as an entity with a balance sheet), its value for the clearing payment vector can safely be ignored. Other authors (see e.g. \cite{Glasserman2016}) choose to exclude the sink node entirely, in this application, however, it is needed for calculation purposes, in the manner introduced here, as we shall see later.}:

\begin{equation}
f(p) = D(p) \left(rC\left(D\left(p\right) f\left(p\right) + \left(I - D\left(p\right)\right)l\right)+r_a a\right) + \left(I - D(p)\right) l
\end{equation}

This map returns for any payment vector a new payment vector, where each bank that is in default under the given payment vector returns the remaining value of its assets. A \textbf{clearing payment vector} $p$ is now any fixed point of this map: $f(p) = p$. This corresponds to the extension of the original model by \cite{Eisenberg2001} introduced by \cite{Rogers2013}.

\textbf{Solution}

In order to solve the model, first fix $D(p)=D$ and solve for the fixed point:

\[ f = rDCDf + rDCl - rDCDl + r_aDa + l - Dl \]
\[ (I - rDCD)(f - l) = D(r_aa + rCl - l) \]
\begin{equation} \label{eqn_ENsolution}
f = (I - rDCD)^{-1} D(r_aa + rCl - l) + l
\end{equation}

The existence of a solution to equations essentially similar to \ref{eqn_ENsolution} has been shown by various authors in the literature (\cite{Eisenberg2001,Rogers2013}), albeit for a slightly different defitionion for $D$ (without setting the sink node to defaulted). It follows as a corollary (\ref{thm_Invertbility_Corollary}) of theorem \ref{thm_Invertibility} that this small change does not affect the existence of a solution. As \cite{Elsinger2012} show, this computation has the advantage that the matrix inversion only has to be applied for the subset of rows and columns which correspond to already defaulted nodes, which is an advantage in real-world applications where the number of banks can be high, but the number of defaults often is low. \cite{Eisenberg2001} show that the clearing payment vector is unique under mild conditions\footnote{The conditions define that certain subsets of the financial system need to have a positive value of external assets. I will use $a_i>0 \forall i$ as a sufficient condition.} and can be obtained from the following iteration called the \textbf{fictious default sequence} initiated with $f_0=l$ \cite{Eisenberg2001}:

\begin{equation}
f_{n+1} = (I - rD(f_n)CD(f_n))^{-1} D(f_n) (r_aa + Cl - l) + l
\end{equation}

Going further, I will set the recovery rate for external assets to $r_a=1$ without loss of generality.

\section{Network centrality}

In this section I will show how the solution of the clearing model can be expressed as a simple network centrality measure under conditions that have a clear economic interpretation. Consider the realization of an exogenous shock $s$ which reduces the original value $o$ of the external assets:

\begin{equation}
a = o + s
\end{equation}

In the following steps, I will write the conditions that hold for $s_i \forall i=1 \dots N-1$. For the sink node, I will assume $a=o>0$. 

Assume that $\forall i < N : (Cl)_i < l_i$ and choose $s_i$ from the interval $s_i \in (-o_i,-o_i + l_i - (Cl)_i)$. This implies that $a_i \in (0,l_i - (Cl)_i)$, hence all banks still have positive external assets, but $a_i + (Cl)_i < l_i$, hence they are in fundamental default. \footnote{Note that the inclusion of a sink node is crucial here, as \cite{Eisenberg2001} show that for $a_i>0 \forall i$ not all nodes can be in fundamental default. In this setup, it is the sink node that has positive equity value, but is still considered to be in default under any $D(x)$ by construction.}. $s$ can thus be interpreted as a shock that renders all banks insolvent, while still keeping the value of their external assets positive (thus not violating the conditions of uniqueness for the clearing payment vector \cite{Eisenberg2001}). Let $m\in(0,1)$ and consider e.g. a linear interpolation:

\begin{equation}
s_i=m(l_i-(Cl)_i-o_i)-(1-m)o_i=ml_i-m(Cl)_i-o_i
\end{equation}

Note that since $a_i + (Cl)_i < l_i$ for all banks except the sink node, $D(l)=I$ (using the fact that the sink node is in default by convention). Hence the shock effectively pushes all banks beyond the default threshold where the non-linearity in the \cite{Eisenberg2001}-model occurs. Under these conditions the fictious default sequence converges after the first iteration and $f(l) = p$ is a clearing payment vector and the solution of the clearing model becomes:

\begin{equation} \label{eqn_Solution}
p = (I - rC)^{-1} (a + rCl - l) + l = (I - rC)^{-1} ((r-m)Cl-(1-m)l) + l
\end{equation}

\begin{theorem} \label{thm_Invertibility} \leavevmode

(a) $I - rC$ is invertible for $r \in [0,1)$ if the system does not contain a sink node

(b) $I - rC$ is invertible for $r \in [0,1]$ if the system contains a sink node
\end{theorem}

\begin{proof} \leavevmode

Note that when $\lambda$ and $v$ are eigenvalues and -vectors of $C$, we have:

\begin{equation}
(\lambda I-C)v=0 \Leftrightarrow (I-\frac{1}{\lambda} C)v=0 \Leftrightarrow \det (\lambda I-C)=0
\end{equation}

So $(I - rC)$ is invertible for $r \ne \frac{1}{\lambda}$. Since $r \in [0,1]$ by definition, we need to investigate whether there are eigenvalues $\lambda \ge 1 \Rightarrow \frac{1}{\lambda} \le 1$. The maximum eigenvalue for a non-negative matrix can be obtained through the Collatz-Wielandt formula:

\begin{equation}
\rho (C) = \max_{x : x_i \ge 0 \land \exists x_i >0 } g(x,C) 
\end{equation}
\begin{equation}
g(x,C) = \min_{i : x_i \ne 0} \frac{(xC)_i}{x_i}
\end{equation}

Without a sink node, if all institution have at least some claims, $C$ would be column-stochastic, hence $g(x)_i = 1 \forall i = \rho(C)$. This proves part (a) for the case where all institutions have some claims. In the case this condition does not hold, one or more columns of $C$ consist of only zeros, which is equivalent to the case with a sink node discussed below.

Given that the sink node adds a column of all zeros to $C$, the maximum eigenvalue is attained for $g(x : x_N = 0)$.  Since the remainder of the columns still sum to 1, we obtain $0 < \rho (C) < 1$, hence $\frac{1}{\lambda} > 1$ for all positive eigenvalues. 
\end{proof}

\begin{corollary} \label{thm_Invertbility_Corollary}
Note that for $0 < D < I$, $0 < \rho (DCD) \le \rho (C)$, which shows the existence of a solution to equation \ref{eqn_ENsolution}.
\end{corollary}

A typical form of a systemic risk measure based on a clearing model is to consider the difference between total liabilities and the clearing payment vector (\cite{Glasserman2016}) $\sigma = l - p$. A systemic risk measure based on equation \ref{eqn_Solution} can thus be written as a centrality measure:

\begin{equation} \label{eqn_Centrality}
\sigma = (I - rC)^{-1}\beta
\end{equation}

With $\beta_i = (1-m)l_i - (r-m)(Cl)_i \forall r,m \in (0,1), i<N$ (the framework would also allow for setting bank-specific recovery and interpolation values by taking $r$ and $m$ as diagonal matrices). The functional form of $\sigma$ is equivalent to that of a Katz centrality measure (\cite{Newman2010, Katz1953}). In the standard definition, $C$ is an adjacency matrix and $\beta = \vec{1}$, so in the general case $\sigma$ could be seen as a generalization. In order for $C$ to be an adjacency matrix, we would have to assume that each bank has at most one creditor, s.t. $C_{ij} \in \{ 0,1 \} \forall i,j$. If we further set $r=m$, $\beta$ simplifies to $\beta_i = (1-r)l_i \forall r\in (0,1), i<N$ and we could obtain $\beta = \vec{1}$ through normalization if every bank has at least some liabilities. Under these conditions the solution of the clearing model is a specific form of a \textbf{Katz centrality measure}.

A large systemic shock that immediately renders all banks in the system insolvent is admittedly a strong assumption. We can relax this assumption by choosing $s_i \in (-o_i,-o_i + l_i - (Cp)_i)$ \footnote{In order to find suitable values for $s$, the model has to be solved first. Given that the aim is to have a small shock size, a reasonable approach would be to start by setting $s_i = (-o_i + l_i - (Cl)_i)) - \frac{k}{MaxSteps} (l_i - (Cl)_i)$ for a suitable $MaxSteps = 1000$, e.g., and iterating over $k = 1, 2, \dots$ until $D(p) = I$.}, which implies a less severe minimum shock since $Cp \le Cl$. Under this condition, $a_i + (Cp)_i < l_i \forall i<N$ and hence $\lim_{n \to N} D(f_n) \to I$. I.e. all banks are pushed into default after the algorithm converges (which happens after at most $N$ steps as shown by \cite{Eisenberg2001}). By choosing an analogous interpolation and inserting into $f(p)$, we again obtain a centrality measure (assuming that $I - (r-m)C$ is invertible):

\begin{equation}
s_i=m(l_i-(Cp)_i-o_i)-(1-m)o_i=ml_i-m(Cp)_i-o_i
\end{equation}

\[
p = (I - rC)^{-1} (a + rCl - l) + l = (I - rC)^{-1} (m(l-Cp+rCl) - l) + l
\]
\begin{equation}
p = (I - (r-m)C)^{-1}m(l+rCl) 
\end{equation}

Which for $r=m$ simplifies to $p = rl+r^2Cl$.



\section{Conclusion} \label{sec_Conclusion}

It is often discussed and/or assumed that standard network centrality measures can give insights about the systemic risk contribution of individual banks (\cite{Puhr2014}). At the same time, sophisticated systemic stress testing tools have been designed using contagion models based on a model of interbank clearing (see e.g. \cite{Elsinger2006}). In this paper, I have showed that one common methodology used in this context, the clearing model developed by \cite{Eisenberg2001} (and in particular, the extension by \cite{Rogers2013}), converges to a special form of a Katz centrality measure as a crisis becomes so severe that all banks in the system are pushed into default (the existence of a world outside the system to which there exist financial liabilities is necessary - this is a mild assumption for a banking system, where deposit-taking is a central part of the business model). 

Given this result, I advocate using contagion models instead of centrality measures when analyzing systemic risk (in the case where the required data is available - \cite{Anand2015} provide a good overview of methods that can be used in the absence of granular network data). As \cite{Puhr2014} and others show, different versions of the Katz centrality measure and its cousins have the highest explanatory power for the output of a contagion model. I have shown that this result is due to a mathematical relation between those measures. A key difference is, however, that clearing models allow to interpret the assumptions made and the results obtained from an economic perspective. Such an analysis shows that the assumptions that one is making when using even a highly specific, specially derived form of a Katz centrality are strong and unrealistic for typical applications (most notably that the entire banking system is in fundamental default)\footnote{Another shortcoming of the $\sigma$-measure is that differences in initial capitalization, which are a significant driver of contagion dynamics, are canceled out through the shock. If one wanted to use the $\sigma$ measure for systemic risk analysis, one should consider including another factor like $\frac{s_i}{o_i + (Cl)_i}$ to capture the initial financial health of firm $i$.} So even if one is looking at the best possible centrality measure - as demonstrated both empirically by \cite{Puhr2014} and formally in this study - one is making strong, potentially unfounded assumptions that could be avoided by using a clearing model instead. 

Avenues for further research include performing a similar analysis for other types of clearing/contagion models, such as \cite{Furfine2003,Battiston2012} and/or establishing a consensus model as has been proposed by \cite{Barucca2016}.

\section*{References}\label{sec_References}
\bibliographystyle{apalike}
\bibliography{Mendeley}

\end{document}